\documentclass[a4paper]{article}

\usepackage{amsmath}
\usepackage{amssymb}
\usepackage{amsthm}
\usepackage{mathrsfs}

\usepackage{ifthen}

\usepackage[usenames,dvipsnames]{color}
\usepackage[pdfcenterwindow, a4paper, colorlinks=true, pagebackref=false]{hyperref}

\newcommand{\titlename}{
On the local energy decay of solutions \\ of the Dirac equation  \\
in the non-extreme Kerr-Newman metric, I}

\hypersetup{
   pdftitle = {\titlename},
   pdfauthor = {M. Winklmeier, O. Yamada},
   anchorcolor = red,
   linkcolor = Maroon,
   citecolor = OliveGreen,
   urlcolor = Sepia,
}

\DeclareMathOperator{\supp}{supp}
\DeclareMathOperator{\sign}{sign}

\newcommand{\rd}{d}

\newcommand{\Ltwo}{\mathscr L^2}
\newcommand{\whLtwo}{\widehat{\mathscr L}^2}

\newcommand{\Cinfty}{\mathscr C_0^\infty}

\newcommand{\pdiff}[2][]{\frac{\partial #1}{\partial #2}}
\newcommand{\diff}[2][]{\frac{d #1}{d #2}}

\newcommand{\im}{{\rm i}}
\newcommand{\e}{{\rm e}}
\renewcommand{\phi}{\varphi}
\newcommand{\id}{{I}}

\newcommand{\mD}{\mathcal D}
\newcommand{\mS}{\mathcal S}

\newcommand{\FA}{\mathfrak A}	
\newcommand{\FD}{\mathfrak D}	
\newcommand{\FH}{\mathfrak H}	
\newcommand{\Fh}{\mathfrak h}	
\newcommand{\FR}{\mathfrak R}	

\newcommand{\mA}{\mathcal A}	

\newcommand{\C}{\mathbb C}
\newcommand{\N}{\mathbb N}
\newcommand{\R}{\mathbb R}
\newcommand{\Z}{\mathbb Z}

\newcommand{\angLdtdphi}{\angL^{t,\phi}}	
\newcommand{\angL}{\mathfrak L}			
\newcommand{\Domk}{\mathfrak R}	
\newcommand{\Ddtdphi}{\Domk^{t,\phi}}	

\newtheorem{theorem}{Theorem}[section]
\newtheorem{lemma}[theorem]{Lemma}
\newtheorem{proposition}[theorem]{Proposition}

\newtheoremstyle{myRemark}{8.0pt plus 2.0pt minus 4.0pt}{8.0pt plus 2.0pt minus 4.0pt}{}{}%
   {\bfseries}
   {.}
   {5pt plus 1.0pt minus 1.0pt}
   {}

\theoremstyle{myRemark}
\newtheorem{remark}[theorem]{Remark}
\newtheorem{definition}[theorem]{Definition}

\begin{document}

\title{\titlename}
\date{\today} 
\author{%
M. Winklmeier\footnote{winklmeier@math.unibe.ch},
O. Yamada\footnote{yamadaos@se.ritsumei.ac.jp}
}

\maketitle

\begin{abstract} 
   \noindent
   \the\textwidth
   We investigate the local energy decay of solutions of the Dirac equation in the non-extreme Kerr-Newman metric.
   First, we write the Dirac equation as a Cauchy problem and define the Dirac operator.
   It is shown that the Dirac operator is selfadjoint in a suitable Hilbert space.
   With the RAGE theorem, we show that for each particle its energy located in any compact region outside of the event horizon of the Kerr-Newman black hole decays in the time mean.
\end{abstract}

\section{Introduction}

The Kerr-Newman metric is the most general stationary solution of Einstein's field equations and has the physical interpretation of a massive charged rotating black hole.
The physical parameters $M$, $Q$ and $a=J/M$ have the interpretation as mass, electric charge and angular momentum parameter.

As in flat spacetime, a spin-$\frac{1}{2}$ particle is described by a Dirac equation.
In the Kerr-Newman metric, the Dirac equation
is given by the coupled system of partial differential equations (see, e.g., \cite{page}, \cite{chandrasekhar}) 
\begin{align}\label{eq:DE:coupled}
   (\widehat\FR+\widehat\FA)\widehat\Psi = 0
\end{align}
where
\begin{align*}
   \widehat\FR\ :=\ \begin{pmatrix}
      \im mr& 0                       & \sqrt{\Delta}\Ddtdphi_+ & 0\\
      0& -\im mr                      & 0  & \sqrt{\Delta}\Ddtdphi_-\\
      \sqrt{\Delta}\Ddtdphi_-  & 0    & -\im mr  & 0\\
      0& \sqrt{\Delta}\Ddtdphi_+      & 0 & \im mr
   \end{pmatrix},
\end{align*}
\begin{align}\label{eq:angular:widehat}
   \widehat\FA\ :=\ \begin{pmatrix}
      -\FD& 0 & 0& \angLdtdphi_+\\
      0& \FD& -\angLdtdphi_-& 0 \\
      0& \angLdtdphi_+& -\FD& 0 \\
      -\angLdtdphi_- & 0& 0& \FD
   \end{pmatrix}
\end{align}
and 
\begin{align}
   \FD\ &:=\ am\cos\theta,\\
   \label{eq:RadialOperator}
   \Ddtdphi_\pm\ &:=\ \pdiff{r}\pm\frac{\im}{\Delta}\left[
      (r^2+a^2)\,\im \pdiff{t} + a\, \im \pdiff{\phi}+ eQr\right]\
   &&\text{on } (r_+,\infty),\\
   \angLdtdphi_\pm\ &:=\ \pdiff{\theta}+\frac{\cot\theta}{2}\mp\left[
      a\sin\theta\ \im \pdiff{t} + \frac{1}{\sin\theta}\ \im \pdiff{\phi}\right]
   &&\text{on } (0,\pi).
\end{align}

The physical parameters 
$m$ and $e$ are the mass and the electric charge of the spin-$\frac{1}{2}$-particle which is described by the wave function $\widehat\Psi$.

The functions $\Delta$ and $\Sigma$ are given by
\begin{align*}
   \Delta(r) \, &=\, r^2 - 2Mr + a^2 + Q^2
   \, =\, (r-M)^2 + a^2 + Q^2 - M^2 \, , \\
   \Sigma(r, \theta) \, &=\, r^2 + a^2 \cos^2\theta.
\end{align*}

According to the black hole parameters $M$, $Q$ and $a$, three cases can arise:
If $a^2 + Q^2 - M^2 < 0$, then the function $\Delta$ has exactly two distinct zeros; this case is called the \emph{non-extreme Kerr-Newman metric}.
If $a^2 + Q^2 - M^2 = 0$, then the function $\Delta$ has exactly one zero; this case is referred to as the \emph{extreme Kerr-Newman metric}.
In both cases the metric is interpreted physically as the spacetime generated by a charged rotating massive black hole.
If $a^2 + Q^2 - M^2 > 0$, then the function $\Delta$ has no zeros; in this case the Kerr-Newman metric has no interpretation as the metric of a black hole.

\medskip
It is a remarkable fact that the Dirac equation~\eqref{eq:DE:coupled} can be separated into ordinary differential equations.
This was first shown by Chandrasekhar (\cite{Cha76}).
The first step is to separate the dependence on the coordinates $t$ and $\phi$ by the ansatz
\begin{align}\label{eq:psi:r:theta:phi}
   \widehat\Psi(r,\theta,\phi, t)
   \, =\, 
   \e^{-\im \omega t} \widetilde\Psi(r,\theta,\phi),
\end{align}
$\omega \in\R$. 
We call $\omega$ an \emph{energy eigenvalue} of equation~\eqref{eq:DE:coupled} if it has a solution of the form~\eqref{eq:psi:r:theta:phi} which is square integrable in the sense explained below in section~\ref{sec:trafo}.%

The separation of the Dirac equation into ordinary differential equations is achieved by the ansatz
\begin{align*}
   \widehat\Psi(r,\theta,\phi,t) 
   \, = \,
   \e^{-\im\omega t} \e^{-\im\kappa\phi} 
   \begin{pmatrix}
      X_-(r) S_+(\theta) \\ X_+(r) S_-(\theta) \\
      X_+(r) S_+(\theta) \\ X_-(r) S_-(\theta)
   \end{pmatrix}\, 
\end{align*}
with $\kappa\in \Z+ \frac{1}{2}$.
One obtains two differential equations, the so-called radial equation for the radial coordinate $r$ and the so-called angular equation for the angular coordinate $\theta$.
Both equations can be written as eigenvalue problems for selfadjoint operators in appropriate Hilbert spaces.
The spectrum of the corresponding angular operator consists of simple and discrete eigenvalues which are unbounded from below and above (see, e.g., \cite{Wi05}).
Belgiorno and Martellini~\cite{BeMa99} showed that the spectrum of the radial equation comprises all of the real axis. 
Recently, it was shown that in the non-extreme Kerr-Newman metric no embedded eigenvalues exist.
In the extreme Kerr-Newman case, however, embedded eigenvalues can exist (see~\cite{Schmid} and \cite{WY06}).

\medskip
In this paper we investigate the local energy decay of solutions of the Dirac equation~\eqref{eq:DE:coupled}.
We assume that the non-extreme case holds, that is, there are no energy eigenvalues.
We do not employ the separation ansatz described above, but we rewrite the Dirac equation as a Cauchy problem.
After 
some transformations, we obtain the equation
\begin{align}
   \label{eq:introEvolution}
   \im\, \pdiff{t}
   \Psi(x,\theta,\phi, t)
   \ =\
   \mS(x,\theta)^{-1} \FH\Psi (x,\theta,\phi,t)
\end{align}
where $\mS(x,\theta)$ is a bounded and boundedly invertible $4\times 4$-matrix and $\FH$ is the formal Dirac operator associated with the Dirac equation \eqref{eq:DE:coupled}.
We show that the expression $\FH$ on the right hand side has a selfadjoint realisation $H$ in an appropriate Hilbert space (Theorem~\ref{thm:Cauchy}).
The operator $H$ can be written as the sum over partial wave operators $H^{(\kappa)}$, $\kappa\in\Z$.
In order to apply the RAGE theorem to the Dirac operator, we show in Theorem~\ref{thm:rellichproperty} that it enjoys a Rellich type property.
One of our main results in this paper is Theorem~\ref{thm:decay} which shows the local energy decay of the wave functions in the time mean:
Let $U^{(\kappa)}$ be the group associated with the skew-selfadjoint operator $-\im\, \mS^{-1} H^{(\kappa)}$.
Then, any function $f$ which is subjected to the time evolution given by \eqref{eq:introEvolution}
satisfies
\begin{align*}
      \lim_{T\to\infty}
      \frac{1}{2T}\int_{-T}^T \left[\int_{-R}^R\int_0^\pi
      \left\| U^{(\kappa)}(t) f(x,\theta) \right\|_{\C^4}^2 \,dx\,d\theta\right] \,dt
      \, =\, 0.
   \end{align*}
for every $R > 0$.
Since the expression in the square brackets is related to the energy of the particle described by $f$ at time $t$ in the region $\Omega_0 := (-R,R)\times(0,\pi)$, this result is related to the decay of the particle's energy in~$\Omega_0$.
The above result remains valid if the bounded electric potential $\frac{eQr}{r^2+a^2}$ given by the background metric is substituted by a possibly unbounded potential $q$ satisfying
\begin{align*}
   q\in C^1((r_+,\infty)),
   \quad
   \lim_{r\to\infty} \frac{q(r)}{r}\ \text{ exists,}
   \quad
   q'(r) = O(r),\ \text{ as }\ r\to\infty.
\end{align*}

\medskip
The local energy decay has also been investigated by Finster, Kamran, Smoller and Yau (\cite{FKNS03}). 
They restrict the Dirac operator to an annulus outside the event horizon and impose Dirichlet type boundary conditions so that the restricted operator has purely discrete spectrum.
Hence the propagator associated with it is a sum of projections.
Extending the radii of the annulus to the event horizon and infinity, respectively, the sum of the propagator representation becomes an integral.

Our approach, however, deals directly with the Dirac operator in the exterior of the black hole.
Using the fact that the partial wave operators $H^{(\kappa)}$ have no point spectrum, the local energy decay of wave functions with initially compact support follows by methods from scattering theory.
In particular, we do not impose any additional boundary conditions.


\section{Transformation of the Dirac equation to a\\ Cauchy problem} 
\label{sec:trafo}
To investigate the time evolution of solutions of the time-dependent Dirac equation we rewrite it as a Cauchy problem.

In the following, we use the Pauli matrices to abbreviate notation:
\begin{align*}
    \sigma_1 \, =\, 
    \begin{pmatrix} 0 & 1 \\ 1& 0 \end{pmatrix}, 
    \quad
    \sigma_2 \, =\,
    \begin{pmatrix} 0 & -\im \\ \im & 0 \end{pmatrix},
    \quad
    \sigma_3 \, =\, 
    \begin{pmatrix} 1 & 0 \\ 0 & -1 \end{pmatrix}.
\end{align*}

In addition, let
\begin{align*}
   I_2 \, =\,
   \begin{pmatrix} 1 & 0 \\ 0 & 1 
   \end{pmatrix},
   \quad
   \beta\, =\, 
   \begin{pmatrix} I_2 & 0 \\ 0 & -I_2 
   \end{pmatrix},
   \quad
   \Sigma_j\, =\, 
   \begin{pmatrix} \sigma_j & 0 \\ 0 & \sigma_j 
   \end{pmatrix},\
   j=1,2,3.
\end{align*}

For the $\phi$-coordinate we always assume the boundary condition
$\Psi(r,\theta,\phi) = \Psi(r,\theta, \phi+2\pi)$.
\smallskip

In the rest of the paper we use the notation
\begin{align*}
   \Omega_3 &:= (-\infty,\infty)\times(0,\pi)\times[0,2\pi),\\
   \Omega_2 &:= (-\infty,\infty)\times(0,\pi).
\end{align*}

\begin{theorem}\label{thm:Cauchy}
   Assume that the function $\Delta$ has at least one real zero \textup{(}that is, either the extreme or the non-extreme Kerr-Newman case holds\textup{)} and denote the largest zero by $r_+$.
   Define the new radial coordinate $x\in(-\infty,\, \infty)$ by
   \begin{align}
      \label{eq:rTrafo}
      \diff[x]{r}\, =\, \frac{r^2 +a^2}{\Delta(r)}.
   \end{align}
   To keep notation simple, we often write $r$ instead $r(x)$.

   A function $\widehat\Psi$ is a solution of the Dirac equation \eqref{eq:DE:coupled} if and only if 
   \begin{align}\label{eq:Cauchy}
      \im \pdiff{t} \Psi
      \,=\, \mS^{-1} \FH\, \Psi
   \end{align}
   where
   \begin{align}\label{eq:PsiTrafo}
      \Psi(x,\theta,\phi,t)\, &=\, (\sin\theta)^{-1/2}\widehat\Psi(r(x),\theta,\phi,t),
   \end{align}
   \begin{align*}
      \mS\, &=\, \id + \frac{a\sqrt{\Delta}\sin\theta}{r^2 + a^2}
      \begin{pmatrix} \sigma_2 & \\ & -\sigma_2 \end{pmatrix},
      \\
      \FH\, &=\, \FH_1 + \FH_2,\\
      \FH_1\, &=\, \begin{pmatrix} \Fh_1 & \\ & -\Fh_1 \end{pmatrix}, \quad
      \Fh_1\, =\, 
      -\sigma_3 \im\,\pdiff{r}
      - 
      \frac{\sqrt{\Delta}}{r^2+a^2}
      \left( 
      \sigma_1 \im \pdiff{\theta} 
      + 
      \sigma_2 
      \frac{1}{\sin\theta}
      \im\,\pdiff{\phi}
      \right),
      \\
      \FH_2\, &=\, 
	 - 
	 \frac{1}{r^2+a^2}
	 \begin{pmatrix} I_2 & \\ & I_2 
	 \end{pmatrix}
	 \Bigl(a \im\pdiff{\phi} + \e Qr \Bigr)
	 \\
	 &\phantom{=\ } 
	 - 
	 \frac{m r \sqrt{\Delta}}{r^2+a^2}
	 \begin{pmatrix} & I_2 \\ I_2 & 
	 \end{pmatrix}
	 +
	 \frac{am \sqrt{\Delta}\cos\theta}{r^2+a^2}
	 \begin{pmatrix}  & \im\, I_2 \\ -\im\, I_2  &
	 \end{pmatrix}
      .
   \end{align*}

   Moreover, the operator $H_0$ defined by
   \begin{align}\label{eq:widehatH}
      H_0\Psi\, =\, \FH\Psi,
      \quad
      \mD(H_0) = 
      \Cinfty (\Omega_3)^4
   \end{align}
   is essentially selfadjoint in the Hilbert space $\Ltwo(\Omega_3)^4$ and the operator 
   \begin{align*}
      \mS^{-1} H_0
   \end{align*}
   is essentially selfadjoint in the Hilbert space
   \begin{align*}
      \Ltwo_\mS(\Omega_3)^4 := \Ltwo( \Omega_3 )^4
      \quad\text{ with scalar product }\
      (\ \cdot\ , \ \cdot\ )_{\mS}
   \end{align*}
   given by
   \begin{align*}
   (\Psi,\, \Phi)_{\mS} = 
   (\Psi,\, \mS\Phi) = 
   \int_{\R} \int_0^\pi \int_0^{2\pi}
   \langle \Psi(x,\theta,\phi),\, \mS(x,\theta)\, \Phi(x,\theta,\phi) \rangle_{\C^4}
    \rd x \,\rd\theta \, \rd\phi;
   \end{align*}
   accordingly, the norm on $\Ltwo_\mS(\Omega_3)$ will be denoted by $\|\,\cdot\,\|_\mS$.
\end{theorem}

\smallskip
Let $H$ be the closure of $H_0$.
Then $\mS^{-1}H$ is the closure of $\mS^{-1}H_0$, and we call $\mS^{-1}H$ the \emph{time-independent Dirac operator} in the Kerr-Newman metric.

\begin{remark}
   \label{remark:EquivNorm}%
   Observe that for each $(x,\theta)\in(-\infty,\infty)\times(0,\pi)$ the matrix $\mS(x,\theta)$ is bounded and boundedly invertible since 
   \begin{align*}
      \left| \frac{a\sqrt{\Delta}\sin\theta}{r(x)^2+a^2} \right|
      \, \le\, \frac{ |a| r(x) }{r(x)^2+a^2}
      \, \le\, \frac{ 1 }{2}.
   \end{align*}
   Therefore, the norms on $\Ltwo(\Omega_3)^4$ and $\Ltwo_\mS(\Omega_3)^4$ are equivalent.
\end{remark}

\begin{proof}[Proof of Theorem~\ref{thm:Cauchy}]
   
First we rearrange equation~\eqref{eq:DE:coupled} such that all time derivatives are on the left hand side and all other terms are on the right hand side, thus we obtain
\begin{align*}
   \makebox[30pt][l]{$\displaystyle
   \left[ 
      \begin{pmatrix} & -\im\, \sigma_3\\ \im\, \sigma_3& \end{pmatrix} 
      + \frac{a\sqrt{\Delta}\sin\theta}{r^2+a^2}
      \begin{pmatrix} & \sigma_1 \\ \sigma_1 &
      \end{pmatrix} 
   \right]\, \im\, \pdiff{t} \widehat\Psi(x,\theta,\phi,t)=
   $} &
   \\[1ex]
   & \begin{aligned}[t]
      \Biggl[&
	 \begin{pmatrix} &I_2 \\ I_2 & 
	 \end{pmatrix} 
	 \pdiff{x}
	 +
	 \frac{\sqrt{\Delta}}{r^2+a^2}
	 \begin{pmatrix} & \im\, \sigma_2\\ \im\, \sigma_2 & 
	 \end{pmatrix} 
	 \pdiff{\theta}
	 \\[1ex]
	 & 
	 + \frac{\sqrt{\Delta}}{(r^2+a^2)\sin\theta}
	 \begin{pmatrix} &-\sigma_1\\ -\sigma_1 & 
	 \end{pmatrix} 
	 \im \pdiff{\phi}
	 + 
	 \frac{1}{r^2+a^2}
	 \begin{pmatrix} & \im\, \sigma_3\\ -\im\,\sigma_3 & \\
	 \end{pmatrix} 
	 \left(a \im\pdiff{\phi} + \e Qr\right)
	 \\[1ex]
	 &
	 + \frac{\im m r \sqrt{\Delta}}{r^2+a^2}
	 \begin{pmatrix}  \sigma_3 & \\ & -\sigma_3 
	 \end{pmatrix}
	 +
	 \frac{ am\sqrt{\Delta}\cos\theta}{r^2+a^2}
 	 \begin{pmatrix}  -\sigma_3 & \\ & -\sigma_3
	 \end{pmatrix}
	 \Biggr]\, \widehat\Psi(x,\theta,\phi,t).
      \end{aligned}
\end{align*}

Multiplication from the left by 
$\Bigl(\begin{smallmatrix} &-\im\, \sigma_3\\ \im\, \sigma_3 \end{smallmatrix}\Bigr)$ 
yields 
\begin{align*} 
   \makebox[30pt][l]{
      $\displaystyle
      \left[\id + \frac{a\sqrt{\Delta}\sin\theta}{r^2+a^2}\ 
	 \begin{pmatrix}  \sigma_2\\ & -\sigma_2
	 \end{pmatrix}
      \right] 
      \im \pdiff{t} \widehat\Psi(x,\theta,\phi,t)
      $}
   & \\
   &=\ 
   \begin{aligned}[t]
     \Biggl[ &
	\begin{pmatrix} -\sigma_3 \\ &\sigma_3
	 \end{pmatrix}
	 \im\,\pdiff{x}
	 + 
	 \frac{\sqrt{\Delta}}{r^2+a^2}
	 \begin{pmatrix} -\sigma_1 & \\ &\sigma_1
	 \end{pmatrix} 
	 \im \pdiff{\theta}
	 \\[2ex]
	 & 
	 + 
	 \frac{\sqrt{\Delta}}{(r^2+a^2)\sin\theta}
	 \begin{pmatrix} -\sigma_2 & \\ &\sigma_2
	 \end{pmatrix} 
	 \im\,\pdiff{\phi}
     - 
	 \frac{1}{r^2+a^2}
	 \begin{pmatrix} I_2 & \\ & I_2 
	 \end{pmatrix}
	 \Bigl(a \im\pdiff{\phi} + \e Qr \Bigr)
	 \\[2ex]
	 &
	 - 
	 \frac{m r \sqrt{\Delta}}{r^2+a^2}
	 \begin{pmatrix} & I_2 \\ I_2 & 
	 \end{pmatrix}
	 +
	 \frac{am\sqrt{\Delta}\cos\theta}{r^2+a^2}
	 \begin{pmatrix}  & \im\, I_2 \\ -\im\, I_2  &
	 \end{pmatrix}
	 \Biggr]\, \widehat\Psi(x,\theta,\phi,t)
  \end{aligned}
  \\[2ex]
  &=\ 
  ( \FH_1 + \FH_2) \Psi (x,\theta,\phi,t),
\end{align*}
thus the first assertion of the theorem is proved.

By Remark~\ref{remark:EquivNorm} the norms on $\Ltwo(\Omega_3)^4$ and $\Ltwo_\mS(\Omega_3)^4$ are equivalent.
Hence, $H_0$ is essentially selfadjoint in $\Ltwo(\Omega_3)^4$ if and only if $\mS^{-1}H_0$ is so in $\Ltwo_\mS(\Omega_3)^4$.
Therefore, it suffices to show that $H_0$ is essentially selfadjoint.

Since $\pdiff{\phi}$ with the boundary condition stated at the beginning of this section has the complete system of eigenfunctions $(\e^{-\im\kappa\phi})_{\kappa\in\Z+1/2}$ and all the differential expressions in the equation above commute with $\pdiff\phi$, each $\Psi\in\Ltwo(\Omega_3\times\R)$ has a representation
\begin{align*}
   \Psi(x,\theta,\phi,t)\ =\  
   \sum_{\kappa\in\Z+1/2} 
   \e^{-\im\kappa\phi} \Psi^{(\kappa)}(x,\theta,t).
\end{align*}
$\Psi$ satisfies \eqref{eq:Cauchy} if and only if each $\Psi^{(\kappa)}$ satisfies 
\begin{align}\label{eq:Cauchykappa}
   \im \pdiff{t} \Psi^{(\kappa)}
   \,=\, \mS^{-1} \FH^{(\kappa)}\, \Psi^{(\kappa)}
\end{align}
where $\FH^{(\kappa)}$, $\FH_1^{(\kappa)}$ and $\FH_2^{(\kappa)}$ are obtained from $\FH$, $\FH_1$ and $\FH_2$ by substituting $\im\pdiff\phi$ with $\kappa$.
Moreover, $H_0$ is essentially selfadjoint if and only if each partial wave operator $H_0^{(\kappa)}$ (the restriction of $H_0$ to the partial wave space represented by $\e^{-\im\kappa\phi}$) is essentially selfadjoint.

In Proposition~\ref{prop:selfadjoint} below, it is shown that $h_{1,0}^{(\kappa)}$, defined by
\begin{align}\label{eq:h10} 
   h_{1,0}^{(\kappa)}\Psi^{(\kappa)} 
   = \Fh_1^{(\kappa)}\Psi^{(\kappa)},
   \quad
   \mD(h_{1,0}^{(\kappa)}) = \Cinfty(\Omega_2)^2,
\end{align}
is essentially selfadjoint in $\Ltwo(\Omega_2)^2$ with the usual scalar product.

Since both $\FH_1$ and $\mS^{-1}$ have diagonal structure and  $\FH_2^{(\kappa)}$ defines a bounded selfadjoint operator on $\Ltwo_\mS(\Omega_2)^4$, the assertion on the essential selfadjointness of $H_0^{(\kappa)}$, $\kappa\in\Z$, and therefore of $H_0$ follows.
\end{proof}

\begin{lemma}\label{lemma:Asummary}
   For $\kappa\in\Z+1/2$ let
   \begin{align*}
      \mA_\kappa := 
      \begin{pmatrix} 0 & 1 \\ -1 & 0 
	 \end{pmatrix} \pdiff{\theta}
      + \frac{\kappa}{\sin\theta} 
      \begin{pmatrix} 0 & 1 \\ 1 & 0 \end{pmatrix} 
   \end{align*}
   and define the unitary matrix
   \begin{align*}
      W := \begin{pmatrix} 0 & 1 \\ \im & 0 \end{pmatrix}.
   \end{align*}

   The operator $\mA_\kappa$ is the angular part of the Dirac equation \eqref{eq:DE:coupled} arising from Chandrasekhar's separation process in the case $a=0$.

   \begin{enumerate}

      \item \label{item:Asummary:i}
      For each $\kappa\in\Z+1/2$, the operator $\mA_\kappa$ with domain $\Cinfty((0,\pi))^2$ is essentially selfadjoint in   
      $\Ltwo((0,\pi),\, \rd\theta)^2$.
      We denote its closure again by $\mA_\kappa$. 
      It is compactly invertible and its spectrum consists of simple eigenvalues only, given by
      \begin{align*}
	 \lambda_{m}^\kappa
	 :=\sign(m) \left(|\kappa| -\textstyle\frac{1}{2} + |m|\right), \ \ m = \pm 1,\, \pm 2,\, \cdots. 
      \end{align*}
      We denote the corresponding normalized eigenfunctions by $g_m^\kappa$.

      \item \label{item:Asummary:ii}
      For every $m\in\Z\setminus\{0\}$ we have\ \
      $\lambda_{-m}^\kappa = - \lambda_m^\kappa$ and $\sigma_3 g_m^\kappa = -g_{-m}^\kappa$.

      \item \label{item:Asummary:iii}
      $\displaystyle 
      W\mA_\kappa W^{-1} \, = \, 
      W^{-1}\mA_\kappa W \, = \, \im\, \sigma_3 \mA_{-\kappa}$.

      \item \label{item:Asummary:iv}
      $\displaystyle 
      W^{-1}\Fh_1^{(\kappa)} W 
      \, =\, 
      W^{-1}\sigma_3 W \Bigl(-\im \pdiff{x} \Bigr) - \frac{\sqrt{\Delta}}{r^2+a^2}\, \mA_\kappa
      \, =\, 
      \sigma_3 \Bigl(\im \pdiff{x} \Bigr) - \frac{\sqrt{\Delta}}{r^2+a^2}\, \mA_\kappa $.

   \end{enumerate}
\end{lemma}

\begin{proof}
   For \ref{item:Asummary:i}, we refer to \cite{Wi05}.
   Statements \ref{item:Asummary:iii} and \ref{item:Asummary:iv} are easily verified by direct computation.
   To prove \ref{item:Asummary:ii} we observe that all eigenvalues of $\mA_\kappa$ are simple and that 
   \begin{align*}
      \mA_\kappa\, \sigma_3 g_m^\kappa
      \, &=\,
      \sigma_3 \sigma_3 \mA_\kappa\, \sigma_3 g_m^\kappa
      \, =\,
      \begin{pmatrix} 1 & 0 \\ 0 & -1 \end{pmatrix}
      \begin{pmatrix} 1 & 0 \\ 0 & -1 \end{pmatrix}
      \mA_\kappa \begin{pmatrix} 1 & 0 \\ 0 & -1 \end{pmatrix} g_m^\kappa
      \\[2ex]
      & =\,
      \begin{pmatrix} 1 & 0 \\ 0 & -1 \end{pmatrix}
      (-\mA_\kappa)\, g_m^\kappa
      \, =\,
      -\lambda_m^\kappa \begin{pmatrix} 1 & 0 \\ 0 & -1 \end{pmatrix} g_m^\kappa
      \, =\,
      -\lambda_m^\kappa\sigma_3 g_m^\kappa.
      \, =\,
      \lambda_{-m}^\kappa\sigma_3 g_m^\kappa.
   \end{align*}

\end{proof}

\begin{proposition}\label{prop:selfadjoint}
   The minimal operator $h_{1,0}^{(\kappa)}$ defined in \eqref{eq:h10} by
   \begin{align*}
      h_{1,0}^{(\kappa)} \psi &= \Fh_{1}^{(\kappa)} \psi 
      =
      \left( 
	 -\sigma_3 \im\,\pdiff{r}
	 - \frac{\sqrt{\Delta}}{r^2+a^2}
	 \left( \sigma_1 \im \pdiff{\theta} 
	    + \sigma_2 \frac{1}{\sin\theta} \im\,\pdiff{\phi}
	 \right)
      \right) \psi,
      \\
      \mD(h_{1,0}^{(\kappa)}) &= \Cinfty((-\infty,\infty)\times (0,\pi))^2
   \end{align*}
is essentially selfadjoint in the space $\Ltwo(\Omega_2)^4$.
\end{proposition}

\begin{proof} 
   Obviously, $\Fh_1^{(\kappa)}$ is hermitian and $h_{1,0}^{(\kappa)}$ is symmetric.
   Hence, the lemma is proved if we have shown that $\ker(h_{1,0}^{(\kappa)*} \pm \im) = \{0\}$. 
   The equality $h_{1,0}^{(\kappa)*} \psi = \pm \im\, \psi$ is equivalent to 
   \begin{align*}
      \Bigl[ W^{-1} h_{1,0}^{(\kappa)*} W \Bigr] W^{-1} \psi
      = \pm \im\, W^{-1} \psi
   \end{align*}
   with the matrix $W$ as defined in Lemma~\ref{lemma:Asummary}.

   If we set $\phi := W^{-1} \psi$, then, by Lemma~\ref{lemma:Asummary}, the equation above is equivalent to
   \begin{align}
      \left[
      \sigma_3 \left(-\im \pdiff{x}\right)
      + \frac{\sqrt{\Delta}}{r^2 + a^2}
      \mA_\kappa
      \right]\phi 
      = \pm \im\, \phi(x,\theta) .
      \label{eq:phi}
   \end{align}

   Since the set of eigenfunctions of the operator $\mA_\kappa$, denoted by $g_n^\kappa$ as in Lemma~\ref{lemma:Asummary}, is  complete in $\Ltwo((0,\pi),\, \rd\theta)^2$, we can expand the function $\phi$ in \eqref{eq:phi} as
   \begin{align*}
      \phi(x, \theta)\, =\, \sum_{m\in\Z\setminus\{0\}} \xi_m(x) g_m^\kappa(\theta)
   \end{align*}
   with functions $\xi_m\in\Ltwo((-\infty,\infty), \rd x)$.
   Thus \eqref{eq:phi} yields
   \begin{align*}
      &\pm \im\, \sum_{m\in\Z\setminus\{0\}} 
      \xi_m(x) g_{m}(\theta)\,\\[1ex]
      &\hspace*{10ex}
      \begin{aligned}
         &=\,
         \sum_{m\in\Z\setminus\{0\}} 
	 -\im \pdiff{x}\, \sigma_3\, \xi_m(x)\, g_m^\kappa(\theta)
	 +
         \sum_{m\in\Z\setminus\{0\}} 
         \frac{\sqrt{\Delta}}{r^2 + a^2}\,
         \xi_m(x)\,
	 \mA_\kappa\, g_m^\kappa(\theta)
         \\[1ex]
         &=\, 
	 \sum_{m\in\Z\setminus\{0\}} 
         \im \pdiff{x}\, \xi_m(x)\, g_{-m}(\theta)
         + \frac{\lambda_m^\kappa \sqrt{\Delta}}{r^2 + a^2}\,
         \xi_m(x)\, g_m^\kappa(\theta) \,.
      \end{aligned}
   \end{align*}
   Taking the scalar product with $g_{\mu}(\theta)$ in $\C^2$ and integrating with respect to $\theta$ yields on the left hand side
   \begin{align*}
      \pm \im\, 
      \int_0^\pi \sum_{m\in\Z\setminus\{0\}} \xi_m 
      \big\langle g_m^\kappa(\theta)\, g_\mu(\theta) \big\rangle_{\C^2} \,\rd\theta
      \, &=\, 
      \pm \im\, 
      \sum_{m\in\Z\setminus\{0\}} \xi_m 
      \int_0^\pi \big\langle g_m^\kappa(\theta)\, g_\mu(\theta) \big\rangle_{\C^2} \,\rd\theta \\
      &\, =\, \pm \im\, \xi_\mu
   \end{align*}
   and on the right hand side
   \begin{multline*}
      \int_0^\pi
      \bigg\langle
      \sum_{m\in\Z\setminus\{0\}} 
      \im \pdiff{x} \xi_m(x) g_{-m}(\theta)
      + \frac{\lambda_m^\kappa \sqrt{\Delta}}{r^2 + a^2}\,
      \xi_m(x) g_m^\kappa(\theta),\
      g_\mu(\theta) 
      \bigg\rangle_{\C^2}\, \rd\theta
      \\[1ex]
      \begin{aligned}
      &=\,
      \sum_{m\in\Z\setminus\{0\}} 
      \im \pdiff{x} \xi_m(x) 
      \int_0^\pi
      \langle g_{-m}(\theta),\, g_\mu(\theta) \rangle_{\C^2}\, \rd\theta
      \\ 
      &\hspace{3ex}
      + \sum_{m\in\Z\setminus\{0\}} 
      \frac{\lambda_m^\kappa \sqrt{\Delta}}{r^2 + a^2}\,
      \xi_m(x) 
      \int_0^\pi
      \langle g_{m}(\theta),\, g_\mu(\theta) \rangle_{\C^2}\, \rd\theta 
      \int_0^\pi
      \langle g_{m}(\theta),\, g_\mu(\theta) \rangle_{\C^2}\, \rd\theta 
      \\[1ex]
      &=\,
      \im \pdiff{x} \xi_{-m}(x) 
      + 
      \frac{\lambda_m^\kappa \sqrt{\Delta}}{r^2 + a^2}\,
      \xi_m(x) 
      \,.
      \end{aligned}
   \end{multline*}
   Hence we obtain
   \begin{align}
      \label{eq:xim}
      \pm \im\, \xi_{\mu}(x) 
      \, =\,
      \im\, \pdiff{x} \xi_{-\mu}(x) 
      + 
      \frac{\lambda_\mu^\kappa \sqrt{\Delta}}{r^2 + a^2}\,
      \xi_{\mu}(x) \,.
   \end{align}
   An analogous procedure with $g_{-\mu}$ gives
   \begin{align}
      \label{eq:xi-m}
      \pm \im\, \xi_{-\mu}(x) 
      \, =\,
      \im\, \pdiff{x} \xi_{\mu}(x) 
      + 
      \frac{\lambda_\mu^\kappa \sqrt{\Delta}}{r^2 + a^2}\,
      \xi_{-\mu}(x) \,.
   \end{align}
   Combining equations \eqref{eq:xim} and \eqref{eq:xi-m} we obtain
   \begin{align}
      \newcommand{\neentry}{\frac{\lambda_\mu^\kappa \sqrt{\Delta}}{r^2+a^2}}
      \label{eq:xdiff}
      \pm\im\, 
      \begin{pmatrix} 
	 \xi_\mu \vphantom{\neentry} \\ \xi_{-\mu}\vphantom{\neentry}
      \end{pmatrix} 
      \, =\, 
      \begin{pmatrix} 
	 \neentry & \im\, \diff{x} \\
	 \im\, \diff{x}  & -\neentry
      \end{pmatrix} 
      \begin{pmatrix} 
	 \xi_\mu \vphantom{\neentry} \\ \xi_{-\mu}\vphantom{\neentry}
      \end{pmatrix} \, .
   \end{align}
   The multiplication operator $\frac{\lambda_\mu^\kappa \sqrt{\Delta}}{r^2+a^2}$ is bounded on  $\Ltwo((-\infty, \infty), \rd x)$ and it is well known that the operator 
   $\Bigl(\begin{smallmatrix} 0 & \im\, \diff{x}\\ \im\, \diff{x} & 0 \end{smallmatrix}\Bigr)$ 
   with domain $\Cinfty(-\infty,\infty)^2$ is essentially selfadjoint in $\Ltwo((-\infty, \infty), \rd x)^2$, see, for instance, \cite[Ch.6.B)]{We87} .
    Hence the same is true for 
    $\Bigl(\begin{smallmatrix}  \frac{\lambda_m^\kappa \sqrt{\Delta}}{r^2+a^2}& \im\, \diff{x}\\ \im\, \diff{x} & 
   -\frac{\lambda_\mu^\kappa    \sqrt{\Delta}}{r^2+a^2} \end{smallmatrix}\Bigr)$ and equation \eqref{eq:xdiff} has only the trivial  
    solution in $\Ltwo((-\infty, \infty), \rd x)$.
    This, however, implies that also equation \eqref{eq:phi} has only the trivial solution and the lemma is proved.
\end{proof}
\medskip

Let $H^{(\kappa)}$ be the unique selfadjoint extension of $\FH^{(\kappa)}$ (see \eqref{eq:Cauchykappa}).
The following result is proved in \cite{Schmid}.
\begin{proposition}\label{prop}
For each half-integer $\kappa$, the spectrum $\sigma(H^{(\kappa)})$ coincides with $\R$, and $H^{(\kappa)}$ has no eigenvalues.    
\end{proposition} 


\section{Rellich property} 

\begin{definition} 
   Let $\Omega$ be an open set in $\R^n$ and $T$ a selfadjoint operator on $\Ltwo(\Omega)$. 
   We say that $T$ has the {\em Rellich property}, if the following is true:
   Let $F\subset\mD(T)$ be such that there exists a positive constant $K$ such that 
   \begin{align*}
	 \|f\|^2 + \|Tf\|^2 \leq K, \quad f \in F.
   \end{align*}

   Then $F$ is precompact in $L^2_{loc}(\Omega')$ for every bounded open set $\Omega' \subset \Omega$, that is, every sequence $(f_n)_n \subset F$ has a subsequence $(f_{n_k})_k$ which converges strongly in $\Ltwo(\Omega')$.
   (As usual, we identify functions $f\in F$ with the corresponding functions $f|_{\Omega'}\in\Ltwo(\Omega')$.)
\end{definition}

The classical result that differential operators in $\Ltwo(\R^d)$ have the Rellich property is cited in the appendix.

The aim of this section is to show that $\FH^{(\kappa)}$ has the Rellich property.
First we proof the following technical lemma needed for Theorem~\ref{thm:rellichproperty}.

\begin{lemma}\label{lemma:technical}
   Let $W$ and $\mA_\kappa$, $\kappa\in\Z+1/2$ be defined as in Lemma~\textup{\ref{lemma:Asummary}} and let $H^{(\kappa)}$ be the Dirac operator in the $\kappa$-th partial wave space \textup{(}see Theorem~\textup{\ref{thm:Cauchy})}.

   Then there exists a positive constant $K$ such that 
   \begin{multline}\label{eq:technical}
      \int_\R \int_0^\pi \left[ \sum\limits_{j=1}^2 
      \left\| \pdiff{x} \Psi_j (x,\theta) \right\|_{\C^2}^2
      + \frac{\Delta}{(r^2+a^2)^2}
      \left\| \mA_{\kappa} \Psi_j(x,\theta) \right\|_{\C^2}^2
      \right] \, \rd\theta\, \rd x
      \\
      \ \le\ 
      K \bigl( 
      \| \mS^{-1}H^{(\kappa)} \Psi\|_{\mS}^2 + \| \Psi\|_{\mS}^2 \bigr)
   \end{multline}
   for all $\Psi= {}^t(\Psi_1,\Psi_2) \in\mD(\mS^{-1}H^{(\kappa)})$ where $\Psi_j\in \Ltwo_{\mS}(\Omega_2)^2$, $j=1,\,2$.
\end{lemma}

\begin{proof} 
   First assume $\Psi \in \Cinfty(\Omega_2)^4$. 
   Observe that
   \begin{align*}
      \frac{1}{2} 
      \leq 
      \biggl\|\biggl( 
      \id +\frac{a\sin{\theta}\, \sqrt{\Delta}}{r^2+a^2}\,\beta\,\Sigma_2
      \biggr)^{-1}\biggr\| \leq 2,
      \qquad (r,\theta)\in(-\infty,\infty)\times(0,\pi).
   \end{align*}
   With the notation of Section~\ref{sec:trafo} we write $\mS^{-1}H^{(\kappa)}$ as the sum of an unbounded operator $\mS^{-1}H_1^{(\kappa)}$ and a bounded $\mS^{-1}H_2^{(\kappa)}$ and obtain
   \begin{align}
      \| \mS^{-1} H^{(\kappa)} \Psi\|^2_\mS
      &\,\ge\, 
      \frac{\,1\,}{2}\ \| \mS^{-1} H_1^{(\kappa)} \Psi\|^2_\mS
      - \| \mS^{-1} H_2^{(\kappa)} \Psi\|^2_\mS
      \nonumber
      \\
      &\,=\, \frac{\,1\,}{2} \int_{\R} \int_0^\pi
      \|\mS^{-1/2} H_1^{(\kappa)}\Psi(x,\theta)\|_{\C^4}^2 dx\,d\theta
      - \| \mS^{-1} H_2^{(\kappa)} \Psi\|^2_\mS
      \nonumber
      \\
      &\,\geq\,
      \frac{1}{4} \int_{\R} \int_0^\pi
      \|H_1^{(\kappa)}\Psi(x,\theta)\|_{\C^4}^2 dx\,d\theta 
      - \| \mS^{-1} H_2^{(\kappa)}\|^2\, \|\Psi\|^2_\mS.
      \label{eq:technical:1}
   \end{align}  

   Recall that 
   \begin{align*}
      H_1^{(\kappa)}
      &=  
      \begin{pmatrix} \sigma_3 &  \\  & - \sigma_3 \end{pmatrix} 
      \left(-\im \pdiff{x} \right) 
      +
      \frac{\sqrt{\Delta}}{r^2+a^2} 
      \begin{pmatrix} W^{-1} \mA_\kappa W &  \\ &  -W^{-1} \mA_\kappa W \end{pmatrix}
      \\[2ex]
      &= 
      \begin{pmatrix} \sigma_3 &  \\  & - \sigma_3 \end{pmatrix} 
      \left[
      \left(-\im \pdiff{x} \right) 
      +
      \frac{\sqrt{\Delta}}{r^2+a^2} 
      \begin{pmatrix} \im\, \mA_\kappa  &  \\ &  -\im\ \mA_\kappa \end{pmatrix}
      \right].
   \end{align*}
   by Lemma~\ref{lemma:Asummary}.
   Since the first matrix in the above representation is unitary, we obtain
   
   \begin{multline}
      \int_{\R} \int_0^\pi \|H_1^{(\kappa)} \Psi(x,\theta)\|_{\C^2}^2 \rd x\, \rd\theta\
      \\
      \begin{aligned}[b]
	 &=\ \int_{\R} \int_0^\pi \left\| 
	 -\im \pdiff{x} \Psi(x,\theta) + 
	 \frac{\sqrt{\Delta}}{r^2+a^2} 
	 \begin{pmatrix} \im\, \mA_{\kappa} & \\ & \im\,\mA_{\kappa}
	 \end{pmatrix}
	 \Psi(x,\theta) \right\|_{\C^4}^2 
	 \rd\theta\, \rd x 
	 \\
	 &=\ \sum_{j=1}^2 \int_{\R} \int_0^\pi 
	 \left\|
	 -\im \pdiff{x} \Psi_j(x,\theta)
	 +\frac{\sqrt{\Delta}}{r^2+a^2} \im\,\mA_{\kappa} \Psi_j(x,\theta)
	 \right\|_{\C^2}^2 
	 \rd\theta\, \rd x 
	 \\
	 &=\ \sum_{j=1}^2 
	 \begin{aligned}[t]
	    &\int_{\R} \int_0^\pi 
	    \left[\left\|\pdiff{x }\Psi_j(x,\theta)\right\|_{\C^2}^2
	    + \frac{\Delta}{(r^2+a^2)^2}
	    \|\mA_{\kappa} \Psi_j(x,\theta)\|_{\C^2}^2 \right]
	    \rd\theta\, \rd x 
	    \\
	    & + 
	    \int_{\R} \int_0^\pi 
	    \left\langle 
	    \Psi_j(x,\theta),\ 
	    \left(\pdiff{x} \frac{\sqrt{\Delta}}{r^2+a^2} \right)
	    \mA_{\kappa} \Psi_j(x,\theta) \right\rangle_{\C^2} 
	    \rd\theta\, \rd x 
	 \end{aligned}
      \end{aligned}
      \label{eq:technical:2}
   \end{multline} 
   where in the last step we have used that $-\im\, \partial_x$ is symmetric and $\mA_{\kappa}$ commutes with $-\im\, \partial_x$ and performed integration by parts.
   There are no boundary terms since $\Psi$ has compact support.

   There exists a positive constant $C$  such that 
   \begin{align} 
      \left|\diff{x} \frac{\sqrt{\Delta}}{r^2+a^2}\right|
      &= \left|\frac{dr}{dx}\frac{d}{dr}\frac{\sqrt{\Delta}}{r^2+a^2}\right|
      = \frac{\Delta}{r^2+a^2}
      \left|\frac{d}{dr}\frac{\sqrt{\Delta}}{r^2+a^2}\right| 
      \nonumber
      \\
      & =\frac{\sqrt{\Delta}}{r^2+a^2}
      \left|
      - \frac{2r\Delta}{(r^2+a^2)^2} 
      + \frac{r-M}{r^2+a^2}
      \right|
      \leq C\, \frac{\sqrt{\Delta}}{r^2+a^2} 
      \label{eq:technical:3}
   \end{align}
   since the term in $|\cdot|$ is continuous and tends to zero for $x\to\pm\infty$.
   Moreover, for every $\varepsilon > 0$ 

   \begin{multline}
      \left\langle \Psi_j,\ 
      \left(\pdiff{x} \frac{\sqrt{\Delta}}{r^2+a^2} \right)
      \mA_{\kappa} \Psi_j \right\rangle_{\Ltwo(\Omega_2)^2}
      \\
      \begin{aligned}[b]
	 &\ge
	 - \left\| \Psi_j \right\|_{\Ltwo(\Omega_2)^2}\, 
	 \left\| \left(\pdiff{x} \frac{\sqrt{\Delta}}{r^2+a^2} \right)
	 \mA_{\kappa} \Psi_j \right\|_{\Ltwo(\Omega_2)^2}
	 \\
	 &\ge
	 -\frac{\varepsilon^2}{2}
	 \left\| \Psi_j \right\|_{\Ltwo(\Omega_2)^2}^2
	 -
	 \frac{1}{2\varepsilon^{2}} \left\| \left(\pdiff{x} \frac{\sqrt{\Delta}}{r^2+a^2} \right)
	 \mA_{\kappa} \Psi_j \right\|_{\Ltwo(\Omega_2)^2}^2.
      \end{aligned}
      \label{eq:technical:4}
   \end{multline}

   Hence, from \eqref{eq:technical:1}, \eqref{eq:technical:2}, \eqref{eq:technical:3} and \eqref{eq:technical:4} we obtain
   \begin{multline*}
      \|\mS^{-1}H^{(\kappa)}\Psi\|^2
      \\
      \, \ge\,
      \begin{aligned}[t]
	 &\frac{1}{4}
	 \sum_{j=1}^2 
	 \int_{\R} \int_0^\pi \left[
	 \left\|\pdiff{x}\Psi_j(x,\theta)\right\|_{\C^2}^2
	 + \frac{(1-\frac{C}{2\varepsilon^2})\Delta}{(r^2+a^2)^2}\,
	 \|\mA_{\kappa} \Psi_j(x,\theta)\|_{\C^2}^2
	 \right] \rd\theta\, \rd x 
	 \\
	 &-
	 \frac{\varepsilon^2}{4}
	 \left\| \Psi \right\|_{\Ltwo(\Omega_2)^4}^2
	 -
	 \| \mS^{-1} H_2^{(\kappa)}\|^2\, \|\Psi\|^2_\mS.
      \end{aligned}
   \end{multline*}

   If we choose $\varepsilon$ small enough such that $\frac{C}{2\varepsilon^2} < 1$ and observe that 
   $\left\| \Psi \right\|_{\Ltwo(\Omega_2)^4}^2 
   \le \|\mS^{-1/2}\|^2 \left\| \Psi \right\|_{\mS}^2$
   then we can choose an $K>0$ large enough such that the estimate in the assertion holds.

   For any $\Psi \in \mD(H^{(\kappa)})$ there exists a sequence $(\Psi_n)_n$ such that 
   \begin{align*}
      \Psi_n \rightarrow \Psi,
      \quad H^{(\kappa)} \Psi_n \rightarrow H^{(\kappa)} \Psi 
   \end{align*}
   in $\Ltwo_\mS(\Omega_2)$.
   Notice that for this sequence also the left hand side of \eqref{eq:technical} converges by the dominated convergence theorem.
   Therefore, assertion holds for all $\Psi\in\mD(H^{(\kappa)})$.
\end{proof} 

Note that $K$ does not depend on $\Psi$ but only on the $\mathscr L_\infty$-bound $C$ of $\displaystyle \sqrt{\Delta} \diff{r}\frac{\sqrt{\Delta}}{r+a^2}$.

\begin{theorem}[Rellich property]
   \label{thm:rellichproperty}
   Let $\kappa$ be any half integer, $R>0$ and set  
   \begin{equation}
   \Omega_{2,R} := (-R,R)\times (-\pi,\pi),
   \end{equation}
   Let $(\Psi_n)_{n}\subset\mD(H^{(\kappa)})$ such that
   \begin{align*}
      \|\Psi_n\|_\mS + \|H^{(\kappa)}\Psi_n\|_\mS \,\le\, K_0,
      \qquad n\in \N,
   \end{align*}
   for some constant $K_0 > 0$.
   Then there exists a subsequence $(\Psi_{n_\ell})_\ell$ and a $\Phi\in\Ltwo( \Omega_{2,R})^4$ such that
   \begin{align*}
      \Psi_{n_\ell} \rightarrow \Phi
      \quad \text{as}\quad \ell \rightarrow \infty 
   \end{align*}
   in $\Ltwo_\mS(\Omega_{2,R})^4$.
\end{theorem}

\begin{remark}
   As usual, we identify elements $\Psi\in\Ltwo(\Omega_2)^4$ with elements of $\Psi\in\Ltwo(\Omega_{2,R})^4$ by restriction.
   The assertion of the Theorem implies that there exists a subsequence $(\Psi_{n_\ell})_\ell$ and a $\Phi\in\Ltwo( \Omega_{2} )^4$ such that
   $
      \chi_R\Psi_{n_\ell} \rightarrow \chi_R\Phi
      \quad \text{as}\quad \ell \rightarrow \infty 
      $
   in $\Ltwo_\mS(\Omega_{2})^4$.
\end{remark}

\begin{proof}[Proof of Theorem~\upshape{\ref{thm:rellichproperty}}]
   Since 
   the norms on $\Ltwo_\mS(\Omega_2)^4$ and $\Ltwo(\Omega_2)^4$ are equivalent, it suffices to consider convergence in the latter space.
   Let $\Psi_n =: (\psi^{1}_n,\, \psi^{2}_n)$ with 
   $\psi_n^j\in\Ltwo(\Omega_2)^2$, $j=1,\,2$.
   Lemma \ref{lemma:technical} and the existence of a positive constant $\delta > 0$ such that 
   \begin{align*}
      \frac{\Delta}{(r^2+a^2)^2}\, \geq\, \delta,
      \quad x \in (-R,R)
   \end{align*}
   imply that there is an $K>0$ such that
   \begin{align} 
      \sum\limits_{j=1}^2
      \int_{-R}^R \int_0^\pi 
      \left( 
      \left\|\pdiff{x}\psi^{j}_n(x,\theta)\right\|_{\C^2}^2 
      + \|\mA_\kappa \psi^{j}_n(x,\theta)\|_{\C^2}^2
      \right) dx\,d\theta
      \hspace{8ex}
      \nonumber
      \\
      \leq K ( \|\Psi_n\|_\mS^2 + \|H^{(\kappa)}\Psi_n\|_\mS^2)
      \leq K K_0^2,
      \label{boundedness}
   \end{align}
   for $n\in\N$.
   Since $\Ltwo((0,\pi),\, \rd\theta)^2$ has a basis of orthonormal eigenfunctions $(g_m^\kappa)_m$ of $\mA_\kappa$ with corresponding eigenvalues
   $\lambda_{m}^\kappa$, $m\in\Z\setminus\{0\}$, see Lemma~\ref{lemma:Asummary}.
   Hence $\psi^1_n$ can be expanded in the double Fourier series
   \begin{align*}
      u_n &= \sum_{\nu,m} 
      \alpha_{\nu,m,n} \e^{-i\nu \pi x/R} g_m^\kappa (\theta),
      \\
      \alpha_{\nu,m,n} &= \frac{1}{2R}
      \int_{-R}^R \int_0^\pi 
      \langle \psi^1_n(x,\theta),\,
      \e^{-\im\nu\pi x/R} g_m^\kappa (\theta) \rangle_{\C^2}\,dx\,d\theta 
      \\  
      \sum_{\nu,m} |\alpha_{\nu,m,n}|^2 
      &= \frac{1}{2R} \int_{-R}^R \int_0^\pi |\psi^1_n|^2 dx\,d\theta 
      \leq \frac{1}{2R} \|\Psi_0\|^2_\mS.  
   \end{align*}
   Moreover, inequality \eqref{boundedness} yields       
   \begin{align*}
      \sum_{\nu,m} 
      \left[(\nu\pi/R)^2 + (\lambda_m^\kappa)^2 \right] |\alpha_{\nu,m,n}|^2
      \leq K K_0.
   \end{align*}
   For fixed $(\nu,m)$ the sequence $(\alpha_{\nu,m,n})_n$ is a bounded sequence, hence it contains a subsequence such that $(\alpha_{\nu,m,n_\ell})_{n_\ell}$ is a Cauchy sequence for any $(\nu,m)$ by a diagonal series argument. 
   Thus for arbitray $L>0$ it follows that
   \begin{multline*}
      \frac{1}{2R}\int_{-R}^R \int_0^\pi
      |\psi^1_{n_\ell} - \psi^1_{n_{j}}|^2 \,dx\,d\theta
      \\
      \begin{aligned}
	 =\ &\frac{1}{2R}\int_{-R}^R \int_0^\pi 
	 |\psi^1_{n_\ell} - \psi^1_{n_{j}}|^2 \,dx\,d\theta
	 \,=\, \sum_{\nu,m} |\alpha_{\nu,m,n_\ell}-\alpha_{\nu,m,n_j}|^2 
	 \\
	 =\ &\sum_{(\nu\pi/R)^2+(\lambda_m^\kappa)^2 \leq L} |\alpha_{\nu,m,n_\ell}-\alpha_{\nu,m,n_j}|^2 
	 \\
	 &+ \sum_{(\nu\pi/R)^2+(\lambda_m^\kappa)^2 \geq L}
	 \frac{
	 \left[(\nu\pi/R)^2+(\lambda_m^\kappa)^2\right]|\alpha_{\nu,m,n_\ell}-\alpha_{\nu,m,n_j}|^2
	 }{(\nu\pi/R)^2+(\lambda_m^\kappa)^2}
	 \\
	 \leq & \sum_{(\nu\pi/R)^2+(\lambda_m^\kappa)^2 \leq L} |\alpha_{\nu,m,n_\ell}-\alpha_{\nu,m,n_j}|^2
	 + \frac{2KK_0}{1+L},
      \end{aligned}
   \end{multline*}         
   which shows that $(\psi_{n_\ell})_\ell$ is a Cauchy sequence in $\Ltwo(\Omega_2)^2$. 

   With the same argument we find that also the sequence $(\psi^2_{n_\ell})_\ell$, contains a convergent subsequence in $\Ltwo(\Omega_2)^2$.        
\end{proof}


\section{Weak local energy decay} 

Since $H^{(\kappa)}$, $\kappa\in\Z\setminus\{0\}$ is selfadjoint in $\Ltwo_\mS(\Omega_2)^4$, it is generator of a unitary group
\begin{align}
   \label{eq:U}
   U^{(\kappa)}(t) := \exp(-\im t \mS^{-1}H^{(\kappa)}).
\end{align}

\begin{theorem}
   \label{thm:exp}
   Let $\kappa\in\Z\setminus\{0\}$, $R>0$ and $\Psi_0\in\mD(H^{(\kappa)})$.
   For any sequence $(t_n)_n \subset \R$ there exist a subsequence $(t_{n_\ell})_\ell$ and  $\Phi \in \Ltwo( \Omega_2)^4$ such that 
   \begin{align*}
      U^{(\kappa)}(t_{n_\ell})\Psi_0 
      \rightarrow \Phi
      \quad\text{as}\quad \ell \rightarrow \infty 
   \end{align*}
   in $\Ltwo(\Omega_2)^4$.
\end{theorem}

\begin{proof} 
   Let $\Psi_n := U^{(\kappa)}(t_n )\Psi_0$.
   Then, for all $n\in\N$,
   \begin{align*}
      \|\Psi_n\|_\mS + \| \mS^{-1}H^{(\kappa)} \Psi_n\|_\mS
      &= \|U^{(\kappa)}(t_n)\Psi_0\|_\mS 
      + \| \mS^{-1}H^{(\kappa)}U^{(\kappa)}(t_n) \Psi_0 \|_\mS
      \\
      &= \|U^{(\kappa)}(t_n)\Psi_0\|_\mS 
      + \| U^{(\kappa)}(t_n)\mS^{-1}H^{(\kappa)} \Psi_0 \|_\mS
      \\
      &= \|\Psi_0\|_\mS + \| \mS^{-1}H^{(\kappa)} \Psi_0 \|_\mS\,.
   \end{align*}
   Thus the assertion follows by Theorem~\ref{thm:rellichproperty} with $K_0 = \|\Psi_0\| + \| H^{(\kappa)} \Psi_0 \|$.
\end{proof}

\begin{theorem}
   \label{thm:decay}
   For every $R>0$ and $\Phi\in\Ltwo_\mS(\Omega_2)$  
   the time mean of the localization of $\Phi$ in $(-R,R)\times(0,\pi)$ is to zero: 
   \begin{align*}
      \lim_{T\to\infty}
      \frac{1}{2T}\int_{-T}^T \left[\int_{-R}^R\int_0^\pi
      \left\| U^{(\kappa)}(t) \Phi(x,\theta) \right\|_{\C^4}^2 \,dx\,d\theta\right] \,dt
      \, =\, 0.
   \end{align*}
\end{theorem}

\begin{proof}
   It is well-known that in the non-extreme Kerr-Newman case the operator $H^{(\kappa)}$ has no eigenvalues (\cite{Schmid}, \cite{WY06}, see also Proposition~\ref{prop}).
   Let $\chi_R(x)$ be the characteristic function of $(-R,R)$.  
   Then, the Rellich property Theorem~\ref{thm:rellichproperty} implies that $\chi_R(H^{(\kappa)}+i)^{-1}$ is a compact operator.
   Therefore, the assertion follows from the RAGE theorem, see \cite{RS79}.
%
\end{proof}   

The decay of the partial waves implies the decay of the solutions of the original Cauchy problem~\eqref{eq:Cauchy}:
   
\begin{theorem}
   Let $\Psi\in\Ltwo_\mS(\Omega_3)^4$.
   Its Fourier expansion 
   \begin{align*}
      \Psi(x,\theta,\varphi)
      = \sum_{\kappa \in \Z+(1/2)} 
      \e^{-i\kappa \varphi}\Psi_\kappa(x,\theta)
   \end{align*}
   converges strongly in $\Ltwo_\mS(\Omega_3)^4$ and 
   \begin{align*}
      \sum_{\kappa \in \Z+(1/2)} \|\Psi_\kappa(x,\theta)\|^2_\mS 
      < \infty.
   \end{align*}
   Then
   \begin{align*}
      \widetilde\Psi(x,\theta,\varphi,t) 
      = \sum_{\kappa \in \Z+(1/2)} 
      \e^{-i\kappa \varphi} U^{(\kappa)}(t) \Psi_\kappa(x,\theta)
   \end{align*}
   is the unique weak solution of 
   \begin{align*}
      \im\pdiff{t}\widetilde\Psi\, =\, \mS^{-1} H \widetilde\Psi,
      \qquad
      \widetilde\Psi(0)\,=\, \Psi
   \end{align*}
   \textup{(}cf. \eqref{eq:Cauchy}\textup{)} in $\Ltwo_\mS(\Omega_3)^4$ and satisfies
   \begin{align*}
      \lim_{T\to\infty} 
      \frac{1}{2T} \int_{-T}^T \left[\int_{-R}^R\int_0^\pi\int_0^{2\pi}
      \left\| \mS^{-1/2}\widetilde\Psi(x,\theta,\varphi,t) \right\|_{\C^4}^2 dx\,d\theta\,d\varphi \right]\,dt \, =\, 0
   \end{align*}
for any $R>0$.     
\end{theorem}

\begin{remark}
   The above results remain still valid if the electric potential $eQr$ in \eqref{eq:RadialOperator} is substituted by any real-valued $C^1$-function $q(r)$ defined on $[r_+,\infty)$ such that $\displaystyle{\lim_{r\rightarrow\infty}\frac{q(r)}{r}}$ exists and $q'(r)=O(1)$ as $r \rightarrow \infty$.
\end{remark}


\appendix 

\section{Equivalent representations of the\\ time-independent Dirac operator} 
\label{sec:alternativeDO}

\begin{remark}
   If we do not apply the transformations \eqref{eq:rTrafo} and \eqref{eq:PsiTrafo} to the wave function $\widehat\Psi$, the time-independent Dirac equation takes the form
   \begin{align}\label{eq:Cauchy1}
      \im \pdiff{t} \widehat\Psi
      \,=\, \mS^{-1} \widehat\FH\, \widehat\Psi
   \end{align}
   with $\widehat\FH\, =\, \widehat\FH_1 + \FH_2$ with $\FH_2$ as in Theorem~\ref{thm:Cauchy} and

   \begin{align*}
      \widehat\FH_1
      \, &=\, 
      \begin{pmatrix} \widehat\Fh_1 & \\ & -\widehat\Fh_1 \end{pmatrix}, \\
      \widehat\Fh_1
      \, &=\, 
      -\sigma_3 \frac{\Delta}{r^2+a^2} \im\,\pdiff{r}
      - 
      \frac{\sqrt{\Delta}}{r^2+a^2}
      \left( 
      \sigma_1 
      \im \Bigl( \pdiff{\theta} + \frac{\cot\theta}{2} \Bigr)
      + 
      \sigma_2 
      \frac{1}{\sin\theta}
      \im\,\pdiff{\phi}
      \right).
   \end{align*}
   $\widehat\FH$ is formally symmetric in the Hilbert space
   \begin{align*}
      \whLtwo_\mS := \Ltwo( (r_+,\infty)\times(0,\pi)\times[0,2\pi)) 
   \end{align*}
   with the scalar product
   \begin{align*}
      (\widehat\Psi,\, \widehat\Phi)_{\mS}^{\wedge} = 
   \int_{r_+}^\infty \int_0^\pi \int_0^{2\pi}
   \langle \widehat\Psi(r,\theta,\phi),\, \mS(r,\theta,\phi)\, \widehat\Phi(r,\theta,\phi) \rangle_{\C^4}
   \frac{\Delta}{r^2+a^2}\rd r\, \sin\theta\,\rd\theta\, \rd\phi.
   \end{align*}

\end{remark}

\begin{remark}
   In Theorem \ref{thm:Cauchy} we have equipped the Hilbert space $\Ltwo((-\infty,\infty)\times(0,\pi)\times(0,2\pi))$ with a scalar product
   $(\,\cdot\,,\,\cdot\,)_\mS$
   such that $\FH$ is formally symmetric.
   Alternatively, we can work with the usual scalar product, but then we have to define the Dirac operator in the following way:
   \begin{align*}
      \widetilde H_0 \Psi\, =\, \mS^{-1/2} \FH \mS^{-1/2} \Psi,
      \quad
      \mD(\widetilde H_0) = 
      \Cinfty ((-\infty,\infty) \times (0,\pi) \times [0, 2\pi) )^4
   \end{align*}
   on $\Ltwo( (-\infty,\infty) \times (0,\pi) \times [0,\pi) )^4$
   equipped with the usual scalar product
   \begin{align*}
      (\Psi,\, \Phi)\widetilde{} = 
      \int_{\R} \int_0^\pi \int_0^{2\pi}
      \langle \Psi(x,\theta,\phi),\, \Phi(x,\theta,\phi) \rangle_{\C^4}
      \rd x \rd\theta \, \rd\phi.
   \end{align*}

   Note that $\mS$ leaves the space $\Cinfty((-\infty,\, \infty)\times(0,\pi)\times[0,2\pi))^4$ invariant.

\end{remark}


\section{Rellich's theorem}
\begin{theorem}[Rellich's theorem]
   Let $d\in\N$,  $K\subseteq\R^d$ compact and $\Omega$ a bounded open neighbourhood of $K$. 
   Further let $s>t\in\N$. 
   Then for each $c>0$ the set
   \begin{align*}
      F := \{u\in H^s(\R^d)\, :\, 
	 \|u\|_{H^s(\R^d)} \le c,\ \supp(u)\subseteq K\}
   \end{align*}
   is precompact in $H^t(\Omega)$, i.e.,
   for each sequence $(f_n)_n\subset F$ there exists an $f_0\in H^t(\R^d)$ and a subsequence $(f_{n_k})_k$ such that $f_{n_k}\rightarrow f_0$ in $H^t(\Omega)$.
\end{theorem}

\begin{proof}
   See, e.g., \cite{Miz73}.
\end{proof}


\bigskip

\begin{center}
\bf Acknowledgements
\end{center}
M. W. is grateful for the hospitality at Ritsumeikan University, Kusatsu.
The research work on this paper started during her visit at Ritsumeikan University, Kusatsu, supported by Open Research Center Project for Private Universities: matching fund subsidy from MEXT, 2004-2008.


\end{document}